\documentclass[a4paper,UKenglish,cleveref, autoref, thm-restate]{lipics-v2021}

\usepackage{amsmath}
\usepackage{amssymb}
\usepackage{algorithm}
\usepackage[noend]{algpseudocode}
\usepackage{tikz}
\usepackage{bm}
\usepackage{graphicx}
\usepackage{amsfonts}
\usepackage[utf8]{inputenc}
\usepackage{stmaryrd}
\usepackage{url}

\EventEditors{Florin Manea and Alex Simpson}
\EventNoEds{2}
\EventLongTitle{30th EACSL Annual Conference on Computer Science Logic (CSL 2022)}
\EventShortTitle{CSL 2022}
\EventAcronym{CSL}
\EventYear{2022}
\EventDate{February 14--19, 2022}
\EventLocation{G\"{o}ttingen, Germany (Virtual Conference)}
\EventLogo{}
\SeriesVolume{216}
\ArticleNo{7}

\usetikzlibrary{positioning,chains,fit,shapes,calc}
\newcommand{\vars}{\mathit{vars}}
\newcommand{\allterms}{\mathit{ter}}
\newcommand{\terms}{\tau}
\newcommand{\dom}{\mathit{dom}}

\newcommand{\gen}{\ensuremath{\mathit{gen}}}

\newcommand{\compatible}{\mathit{comp}}
\newcommand{\enforce}{\triangleleft}

\newcommand{\img}{\mathit{img}}

\newcommand{\au}{\texttt{au}}

\newcommand{\dau}{\texttt{dau}}

\nolinenumbers

\title{Anti-unification of Unordered Goals}

\author{Gonzague Yernaux}{University of Namur - Faculty of Computer Science (Belgium)\\Namur Digital Institute}{gonzague.yernaux@unamur.be}{https://orcid.org/0000-0001-6430-8168}{}
\author{Wim Vanhoof}{University of Namur - Faculty of Computer Science (Belgium)\\Namur Digital Institute}{wim.vanhoof@unamur.be}{https://orcid.org/0000-0003-3769-6294}{}

\authorrunning{G. Yernaux and W.\, Vanhoof} 

\Copyright{Gonzague Yernaux and Wim Vanhoof}


\category{} 

\begin{document}   

	\maketitle 
	\begin{abstract}
		Anti-unification in logic programming refers to the process of capturing common syntactic structure among given goals, computing a single new goal that is more general called a generalization of the given goals. Finding an arbitrary common generalization for two goals is trivial, but looking for those common generalizations that are either as large as possible (called largest common generalizations) or as specific as possible (called most specific generalizations) is a non-trivial optimization problem, in particular when goals are considered to be \textit{unordered} sets of atoms. In this work we provide an in-depth study of the problem by defining two different generalization relations. We formulate a characterization of what constitutes a most specific generalization in both settings. While these generalizations can be computed in polynomial time, we show that when the number of variables in the generalization needs to be minimized, the problem becomes NP-hard. We subsequently revisit an abstraction of the largest common generalization when anti-unification is based on injective variable renamings, and prove that it can be computed in polynomially bounded time. 
	\end{abstract}


	\section{Motivation and Objectives}

Anti-unification refers to the process of generalizing two (or more) program objects $S$ into a single, more general, program object that captures some of the structure that is common to all the objects in $S$. In a classical logic programming context, the atom $p(X,Y)$ can thus be seen as a generalization of both the atoms $p(f(A), U)$ and $p(f(g(B)),h(C))$, thanks to the variables $X$ and $Y$. 

Anti-unification constitutes a useful tool in various contexts ranging from program analysis techniques (including partial evaluation, refactoring, automatic theorem proving, program transformation, formal verification and test-case generation~\cite{au-applications,calculus-constr,DESCHREYE1999231,lg-gs,under-implication}) to automated reasoning \cite{ilp-theory-and-methods,Muggleton90efficientinduction} or analogy making~\cite{analogy-making}, supercompilation~\cite{Sorensen95analgorithm} and even plagiarism detection~\cite{clones}. Many of these static techniques are executed on programs written in the form of (constraint) Horn clauses, a formalism that has been praised for its ability to capture a program's essence in a quite universal and straightforward manner~\cite{horn-clauses-intermediate-representation}. 

In the introductive example above, the presence of variables $X$ and $Y$ conceptually allows concrete instances (i.e. less general objects) to harbor any value at the positions corresponding to the variable positions. The generalization process is indeed usually achieved by ``forgetting'' parts of the objects to generalize (either by replacing sub-objects with variables or by dropping them altogether): the less syntactic information in an object, the more general it is. Most anti-unification methods are thus steered by a \textit{variabilization} algorithm determining how to ``forget'' object parts when necessary while keeping (common) parts in the generalization. Therefore, in general one is typically interested in computing what is often called a most specific generalization (or synonymously least general generalization), that is a generalization that captures a maximal amount of shared structure. With the atoms of the example above, the common generalization $p(f(X), Y)$ is in that regard a \textit{better} anti-unification result than $p(X,Y)$, as it exhibits more common structure (namely the use of functor $f$). As this example hints, ``better'' results are often obtained at the cost of more complex anti-unification algorithms. In that regard, computing more specific generalizations often boils down to performing some kind of optimization in the variabilization process. 

In a classical approach where goals are \textit{ordered} sequences of atoms, a goal $G$ is more general than some other goal $G'$ if $G'$ can be obtained by applying on $G$ some substitution $\theta$, being a mapping from variables to values. $G$ then typically harbors more variables than $G'$, making it a less instantiated, thus more general, version of $G'$. In that case, $G$ and $G'$ are related by the $\theta$-subsumption relation from~\cite{plotkin}, often considered to be a foundation of Inductive Logic Programming where anti-unification is used as a way to learn a general hypothesis from specific examples~\cite{ilp-theory-and-methods}. As the name may suggest, looking for a generalization that is common to a group of program artefacts (be it terms, atoms, goals or even predicates as a whole) is referred to as anti-unification due to it being the dual operation of unification. Both can, in fact, be applied in similar contexts. Such applications of (anti-)unification include program transformation techniques for partial deduction \cite{Gallagher:1993:TSL:154630.154640,DESCHREYE1999231}, fold/unfold routines \cite{DBLP:journals/csur/PettorossiP98}, invariant generation~\cite{DBLP:conf/synasc/KovacsJ05} and reuse of proofs~\cite{unranked-2-order-au,calculus-constr}. 

The study of anti-unification so far has mainly been focused on such ordered goals. However, many applications require goals to be defined as (\textit{unordered}) sets of atoms. It is the case, for instance, when considering the most declarative semantics of logic programs~\cite{lp-semantics,clp-semantics,horn-clauses-intermediate-representation}. Having a clear overview of anti-unification operators computing most specific generalizations for unordered goals (sometimes called \textit{linear} generalizations) in logic programs is necessary for generalization-driven semantic
clone detection with programs composed of constraint Horn clauses~\cite{clones,DBLP:conf/ppdp/MesnardPV16}. Indeed, generalization operators allow to quantify a certain amount of structural similarity between different predicate definitions by highlighting what parts these have in common. In~\cite{clones}, this quantitative similarity measurement is used as an indication of which semantic-preserving program transformation should be applied next in order to ultimately assess whether two programs (or predicates) are semantic clones. A quite similar approach has already been taken in the case of ordered goals in~\cite{au-applications}, an obvious application of this being plagiarism detection. 

Directing our interest towards unordered goals also has the advantage of broadening the traditional anti-unification theories usually rooted in a setting where logic programming is based on operational semantics, by extending the theories to the more general area of Constraint Logic Programming (CLP), unordered goals being a crucial ingredient of the CLP(X) framework. The fixpoint semantics of CLP programs are indeed typically defined with no regard to the order of appearance of the atoms in a clause's body~\cite{clp-semantics}. While CLP is interesting in its own right, it is also considered a serious candidate for representing abstract \textit{algorithmic knowledge}, rather than mere computations, in a quite universal manner~\cite{horn-clauses-intermediate-representation}. In that regard, focusing on unordered goals could pave the way for performing anti-unification at the algorithmic level rather than at the level of language-specific operations. 

The topic of anti-unification in the case of unordered goals has ocasionally come up in studies focussed on related fields such as \textit{equational} anti-unification, encompassing theories specified by commutativity or associative-commutativity axioms. The topic has been treated for first-order theories~\cite{order-sorted} as well as higher-order variants~\cite{kutsia_2020}. The latter work applies to the first-order case as well and provides polynomial algorithms for variants of anti-unification for unordered input. A grammar-based approach to equational anti-unification including commutative theories, called E-generalization, was introduced in~\cite{e-generalization} and refined with a working implementation in~\cite{e-generalization-improved}. The authors of~\cite{unranked-2-order-au} elaborate a \textit{rigid anti-unification} algorithm that can apply to unordered (and so-called \textit{unranked}) theories by instantiating a parameter called rigidity function, a direct application of which being the computation of longest common substrings. The algorithms described in all of these works can be used to compute what we will call $\sqsubseteq$-common generalizations below in the present paper. Although none of these works develop a general (non-equational) taxonomy allowing to extend the results beyond that simple setting, nor discusses variable- or injectivity-based variants of anti-unification operators, 
their usages do point out other interesting (and recent) applications of anti-unification when focused on unordered goals, namely detection of recursion schemes in functional programs (as explained in~\cite{BARWELL2018669}) and techniques for learning bugfixes from software code repositories (an example being~\cite{rolim2018learning}). 

Anti-unification techniques that are adapted for CLP(X) have been defined in~\cite{gen}, but its focus is set on a polynomial abstraction procedure for a specific case where terms cannot be generalized (only variables can) and where generalization has to be carried out through injective substitutions. 
While~\cite{gen} provides useful insights and results, it lacks a more general and in-depth study of the used generalization operator. In this work we broaden, generalize and complete the latter work by providing a detailed and systematic study of generalization operators and their characteristics in the context of CLP. 

The main contributions of the present work are the following. In Section~\ref{section-preliminaries} we define relations close to the well-known $\theta$-subsumption in an effort of adapting this notion to the case of unordered goals. As will be illustrated throughout the paper, our adaption of anti-unification to unordered goals makes the usual subsumption techniques unusable. In Section~\ref{section-relation-1} we reframe the problem of looking for a most general/largest generalization as an optimization problem, parametrized by the \textit{generalization operator} (or anti-unification strategy) and \textit{variabilization function} (responsible for introducing variables in the resulting generalization) at hand.  We will see that given two unordered goals as input, searching for such generalizations can be done in polynomial time. The algorithms, as well as their worst-case time complexities, are detailed throughout the development of our anti-unification framework. 
 In Section~\ref{section-relation-2} we provide an in-depth examination of several key variations of the anti-unification problem, namely variable generalization (where no terms are allowed to be generalized), injective generalization (where the generalizing substitutions need to be injective) and dataflow optimization (where the number of generalizing variables needs to be minimized) -- the latter of which is proved to make the anti-unification statement NP-hard. Finally, addressing this last problem more in depth in Section~\ref{section-relation-3} we revisit a tractable abstraction that was introduced in~\cite{gen} but we provide for the first time a formal proof of its worst-case complexity, showing that the approximation can effectively be computed in polynomially bounded time. With the exception of this last result, the proofs of propositions, lemmas and theorems are provided in the Appendices.

	\section{Preliminaries}\label{section-preliminaries}
In the following, we introduce concepts and notations that will be used throughout the paper. We suppose a language of Horn clauses defined over a context, which is a 4-tuple $\langle X, \mathcal{V}, \mathcal{F}, \mathcal{Q}\rangle$, where $X$ is a non-empty set of constant values, $\mathcal{V}$ is a set of variable names, $\mathcal{F}$ a set of function names and $\mathcal{Q}$ a set of predicate symbols. The sets $X, \mathcal{V}, \mathcal{F}$ and $\mathcal{Q}$ are all supposed to be disjoint sets. Symbols from $\mathcal{F}$ and $\mathcal{Q}$ have an associated arity (i.e. its number of arguments) and we write $f/n$ to represent a symbol $f$ having arity $n$. Given a context $\mathcal{C} = \langle X, \mathcal{V}, \mathcal{F}, \mathcal{Q}\rangle$, we define the set of \textit{terms} over it as $\mathcal{T}_\mathcal{C}= X \cup \mathcal{V} \cup \{f(t_1, t_2, \dots, t_n) | f/n \in \mathcal{F}\wedge\forall i \in 1..n : t_i \in \mathcal{T}_\mathcal{C}\}$. Terms are thus ground domain constants, variables and functor-based expressions over other terms. In what follows we will use uppercase symbols to represent variables whereas lowercase symbols will be used for function and predicate symbols. 
The set of \textit{atoms} over $\mathcal{C}$ is defined as $\mathcal{A}_\mathcal{C}=\{p(t_1,\ldots,t_n)\:|\:p/n\in \mathcal{Q} \wedge \forall i\in 1..n:t_i\in\mathcal{T}_\mathcal{C}\}$. An atom $p(t_1,\ldots,t_n)$ is understood as representing an atomic formula involving the predicate $p$ over $n$ arguments, the arguments being represented by terms.
%
A \textit{goal} $G$ is a set of atoms, representing an (unordered) conjunction, thus $G\subseteq\mathcal{A}_\mathcal{C}$. 

\begin{example}\label{ex:syntax}
    Let us consider a numerical context (e.g. $X = \mathbb{Z}$ and $\mathcal{F}$ is the set of usual functions over integers composed of addition ($+$), substraction ($-$), integer division ($/$), multiplication ($*$) and modulo ($\%$)). Supposing $X$ and $Y$ to represent variables, then the following are terms: $3$, $X$, $+(3,X)$,  $+(4, *(X,\%(Y, 2)))$.
    Given predicates $p/1$, $q/1$, $r/2$ and $c/2$, the following are atoms: $p(3)$, $q(X)$, $r(+(2,4), +(3,X))$
\end{example}



In what follows we will often leave the underlying context implicit and simply talk about variables, function and predicate symbols. A \textit{substitution} is a mapping from variables to terms and will be denoted by a Greek letter. For any substitution $\sigma : \mathcal{V} \mapsto \mathcal{T}_\mathcal{C}$, $\dom(\sigma)$ represents its domain, $\img(\sigma)$ its image, and for a program expression $e$ (be it a term, an atom or a goal) and a substitution $\sigma$, we write $e\sigma$ to represent the result of substitution application, i.e. simultaneously replacing in $e$ those variables $V$ that are in $\dom(\sigma)$ by $\sigma(V)$. A \textit{renaming} is a special kind of substitution, mapping variables to variables only. Thus for any renaming $\rho$ we have that $\img(\rho) \subseteq \mathcal{V}$. 
We can now define what constitutes a generalization relation $\sqsubseteq$, which essentially defines a goal as more general than another if the latter is a potentially larger and potentially more instantiated goal than the former. 

\begin{definition}
	\label{def-generalization}
	Let $G$ and $G'$ be goals. $G$ is a \emph{generalization} of $G'$ if and only if there exists $\theta$, a substitution such that $G\theta \subseteq G'$. We denote this fact by $G\sqsubseteq G'$ (or sometimes $G\sqsubseteq_\theta G'$ if we want to emphasize the substitution $\theta$ in question).
\end{definition}

\begin{example}
	$\{p(X,Y,Z)\}, \{q(a(X))\}$ and $\{p(t(1), Y, u(Z)), q(W)\}$ are generalizations of $\{p(t(1), t(2), u(+(4,X))), q(a(t(u(1))))\}$.
\end{example} 

In some applications (e.g. for some usual computation domains in Constraint Logic Programming), it makes sense to use a more restricted generalization relation, in which variables are substituted by other variables rather than terms. As such, when the substitution $\theta$ in Definition~\ref{def-generalization} is a renaming, we say that $G$ is a \textit{variable generalization} of $G'$, which we denote by $G\preceq G'$ (or sometimes $G\preceq_\theta G'$ to emphasize the renaming $\theta$ in question). When considering the relation $\preceq$, only variables are generalized and the function symbols are considered as being a part of the language structure itself (i.e. they are not subject to generalization). This can be advantageous, for instance in applications working with a small finite domain such as Booleans, where considering $G = \{=\!\!(A,B)\}$ to be a generalization of both $\{=\!\!(X, true)\}$ and $\{=\!\!(Y, false)\}$ can feel like ignoring too much of the goal's semantics.
%

Our generalization relations are variations of the classical $\theta$-subsumption~\cite{plotkin}, adapted to goals being sets rather than ordered sequences of atoms. They share the following property with $\theta$-subsumption.
\begin{proposition}\label{prop-quasi-order}
	Relations $\sqsubseteq$ and $\preceq$ are quasi-orders. 
\end{proposition}

We will now turn our attention towards the basic concept in anti-unification, namely that of a goal being a \textit{common generalization} of some given goals~\cite{plotkin}. In the following, we restrict ourselves to common generalizations of \textit{two} goals, but the concept can straightforwardly be extended to any number of goals. As for notation, when a result or definition holds for both our relations $\preceq$ and $\sqsubseteq$, for the sake of simplicity we will sometimes use $\leqslant$ to denote both relations at once. 

\begin{definition}
	\label{def-common-generalization} 
	Let $G_1,\dots,G_n$ be goals and $\leqslant$ a generalization relation. Then $G$ is a \emph{$\leqslant$-common generalization} of $\{G_1,\dots,G_n\}$ if and only if $\forall i \in 1..n : G\leqslant G_i$.
\end{definition}

The definition essentially states that each $G_i (1\le i\le n)$ can be generalized by $G$ through its own substitution. Formally there exist $\theta_1, \dots, \theta_n$ such that $\forall i \in 1..n : G \sqsubseteq_{\theta_i} G_i$. A common generalization of goals is thus, in essence, a part of their shared atomic structure, with a possible introduction of variables in certain places -- the liberality of which depends on the underlying relation. Note that renamings being (restricted) substitutions, for any two goals $G$ and $G'$ it holds that $G\preceq_\theta G' \Rightarrow G\sqsubseteq_\theta G'$ so that if a goal is a $\preceq$-common generalization of a set of goals it is also a $\sqsubseteq$-common generalization of said goals. 

\begin{example}\label{ex:common-gen}
	Let $G_1 = \{p(t(X), Y), q(3, f(X))\}$ and $G_2 = \{p(5, Z), q(3, f(Z))\}$. The following is a (non-exhaustive) list of $\sqsubseteq$-common generalizations of $G_1$ and $G_2$: $\emptyset, \{p(V_1, V_2)\}$, $\{q(3, f(V_1))\}, \{p(V_1, V_2), q(V_3, V_4)\}, \{p(V_1, V_2), q(3, V_3)\}$. The following are $\preceq$-common generalizations of $G_1$ and $G_2$ as well: $\emptyset, \{q(3, f(V_1))\}$. 
\end{example} 

As a slight lexical abuse, given atoms $\{A_1,\dot,A_n\}$ we will say that an atom $A$ is a $\leqslant$-common generalization of $\{A_1,\dot,A_n\}$ iff $\{A\}$ is a $\leqslant$-common generalization of $\{A_1,\dot,A_n\}$.
Note that no matter the relation and no matter the goals $G_1$ and $G_2$, at least one common generalization will always exist: the empty goal $\emptyset$. Obviously, wherever possible we are interested in more detailed representations of the common structure found in goals. 

For an expression $e$, we use $\vars(e)$ to represent the set of variables that appear in $e$ and $\terms(e)$ to denote the multiset of all atoms and non-variable terms occurring in $e$. We will sometimes refer to the cardinality of $\tau(e)$ as the $\tau$-value of $e$. The multiset of all atoms and terms, variables included, is denoted by $\allterms(e)$.
\begin{example}
	Let $G$ be the goal $\{p(f(x, Y)), q(Y,X)\}$. The multiset $\tau(G)$ is equal to $\{p(f(x,Y)), f(x,Y), x, q(Y,X)\}$. $G$'s $\tau$-value is $4$, $\vars(G) = \{X,Y\}$ and $\allterms(G)$ is the multiset $\{p(f(x,Y)), f(x,Y), x, Y, q(Y,X), Y, X\}$. 
\end{example} 

 One is typically interested in those common generalizations that are the \textit{most specific}, i.e. that capture as much common structure as possible amongst $G_1$ and $G_2$~\cite{plotkin}.  

\begin{definition}
	\label{def-msg}
	Given goals $G_1,\dots,G_n$ and $G$ such that $G$ is a $\leqslant$-common generalization of $\{G_1,\dots,G_n\}$, we say that $G$ is a \emph{$\leqslant$-most specific generalization ($\leqslant$-msg)} of $\{G_1,\dots,G_n\}$ if $\nexists G'$, another $\leqslant$-common generalization of $\{G_1,\dots,G_n\}$, such that $|\terms(G')|>|\terms(G)|$.
\end{definition}

\begin{example}\label{ex:msg}
	Consider again the goals $G_1$ and $G_2$ from Example~\ref{ex:common-gen}. It is easy to see that $G = \{p(V_1, V_2), q(3, f(V_3))\}$ has a higher $\terms$-value than all the other common generalizations listed in the example; $G$ is in fact a $\sqsubseteq$-msg of $G_1$ and $G_2$, and in this case, all other msg's of $G_1$ and $G_2$ differ from $G$ only in a renaming of the variables $V_1$, $V_2$ and $V_3$. As for relation $\preceq$, the goal $\{q(3, f(V_1))\}$ as well as its variants with $V_1$ renamed are $\preceq$-msg's of $G_1$ and $G_2$. 
\end{example} 

A weaker yet useful measure for comparing common generalizations is the number of atoms (i.e. the cardinality) of the common generalization $G$. 

\begin{definition}
	\label{def-mcg}
		Given goals $G_1,\dots,G_n$ and $G$ such that $G$ is a $\leqslant$-common generalization of $\{G_1,\dots,G_n\}$, we say that $G$ is a \emph{$\leqslant$-largest common generalization ($\leqslant$-lcg)} of $\{G_1,\dots,G_n\}$ if $\nexists G'$, another $\leqslant$-common generalization of $\{G_1,\dots,G_n\}$, such that $|G'|>|G|$.
\end{definition}

\begin{example}\label{ex:lcg} 
	Let us again take a look at the goals from Example~\ref{ex:common-gen}. Each goal of size 2 (such as $\{p(V_1, V_2), q(V_3, V_4)\}$) is a $\sqsubseteq$-lcg, seeing that no larger $\sqsubseteq$-common generalization can exist as $|G_1| = |G_2| = 2$. Regarding the $\preceq$ relation, common generalizations of size 1 (e.g. $\{q(3, f(V_1))\}$) are the largest that exist in the example since the atoms involving $p/2$ have no $\preceq$-common generalization because of the structural difference in their first argument. 
\end{example} 

Before we can dive into the process of computing common generalizations, a few more preliminary observations need to be assessed regarding relations $\sqsubseteq$ and $\preceq$. First, 
%
we state that there is no other way for a common generalization to be most-specific than to harbor as many atoms as possible.
\begin{proposition}\label{prop-msg-lcg}\label{prop-msg-lcg-preceq}
	Any $\leqslant$-msg is a $\leqslant$-lcg and any $\preceq$-lcg is a $\preceq$-msg. 
\end{proposition}	

\begin{example} 
	Let us consider $G_1 = \{a(Y,Z), a(t(1), X)\}$ and $G_2 = \{a(t(1), E)\}$ as well as $G = \{a(t(1), V_1)\}$. It is easy to see that $G$ (and all its variations with $V_1$ renamed) is the only $\preceq$-lcg (thus $\preceq$-msg), as $G_2$'s atom can only be anti-unified with the atom in $G_1$ that has the same structure -- and so the same $\terms$-value. Here, $G$ is also a $\sqsubseteq$-msg (thus a  $\sqsubseteq$-lcg).
\end{example} 

Regarding $\sqsubseteq$, the converse of the above proposition ("any $\sqsubseteq$-lcg is a $\sqsubseteq$-msg") is not true, as shown by the following example.
Let us consider $G_1 = \{a(Y,Z), a(t(1), X)\}$ and $G_2 = \{a(t(1), E)\}$ as well as the following $\sqsubseteq$-lcg's: $G = \{a(V_1,V_2))\}$ and $G' = \{a(t(1), V_1)\}$. Obviously $|\tau(G')| = 3 >|\tau(G)| = 1$. In fact, $G'$ is a $\sqsubseteq$-msg for this example. 

For a set of goals $\{G_1,\dots,G_n\}$, we have defined most specific and largest generalizations using the plural. In fact, by the definitions above and as appears clearly in our examples, $G_1, \dots, G_n$ can have more than one $\preceq$-lcg (and equivalently $\preceq$-msg), but all are equivalent modulo a variable renaming. The same does not necessarily hold with the relation $\sqsubseteq$: 
there might exist more than one sensibly different $\sqsubseteq$-lcg's, depending on the degree at which the different terms are abstracted away through the generalizations process. The following example shows that a similar observation holds for $\sqsubseteq$-msg's. 

\begin{example} 
	Consider the goals $G_1 = \{p(t,u)\}$ and $G_2 = \{p(t,X),p(X,u)\}$. There are two possible structures of $\sqsubseteq$-msg's, namely $\{p(t,V_1)\}$ and $\{p(V_1, u)\}$. There is one more possible structure of $\sqsubseteq$-lcg, namely $\{p(V_1, V_2)\}$
\end{example} 

For the sake of clarity, in the results and discussions that follow we will simplify and consider common generalizations of \textit{two} goals, but the ideas are straightforwardly applicable to groups of more than two goals. Furthermore, when discussing the generalization process of two goals we will suppose that the goals in question share no common variable name. This hypothesis is by no means a loss of generality as renaming all variables from one goal into fresh, unused variable names can ensure this property while not altering the goal's semantics.

	\section{Large and Specific Generalizations}\label{section-relation-1}

In this section we prove that msg's and lcg's as defined above can be computed with polynomial-time algorithms. First, we need the concept of a \textit{variabilization} which is basically a function mapping couples of terms to new variables. 

\begin{definition}
	Given a context $\langle X, \mathcal{V}, \mathcal{F}, \mathcal{Q}\rangle$, let
	$V\subset \mathcal{V}$ denote a set of variables. A function $\Phi_V : \mathcal{T}^2\mapsto\mathcal{V}\cup X$ is called a \emph{variabilization function} if, for any $(t_1, t_2)\in\mathcal{T}^2$ it holds that if $\Phi_V(t_1,t_2) = v$, then
	$(1)\; v \notin V,\;(2)\;\nexists (t_1', t_2') \in \mathcal{T}^2 : (t_1', t_2') \neq (t_1, t_2) \wedge \Phi_V(t_1', t_2') = v, \; (3) \; v \in X \Leftrightarrow t_1 = t_2\in X$ and in that case, $v = t_1 = t_2$.
\end{definition}

Note that a variabilization function $\Phi_V$ introduces a new variable (not present in $V$) for any couple of terms, except when the terms are the same constant. It can thus be seen as a way to introduce new variable names when going through the process of anti-unifying two goals. In what follows, when manipulating goals $G_1$ and $G_2$, we will use $\Phi_{\vars(G_1\cup G_2)}$ to represent an arbitrary variabilization function. If the goals at hand are clearly identified from the context, we will abbreviate the notation to $\Phi$. In most upcoming examples we will use applications of $\Phi$ (e.g. $\Phi(X,Y), \Phi(t(X), 5),\dots$) rather than coined variable names (e.g. $V_1, V_2,\dots$) when an anti-unification operator is -- ostensibly or not -- at work.

Algorithm~\ref{algo-rel-1-lcg} shows the intuitive solution for computing a lcg with two goals $G_1$ and $G_2$ (where we suppose $|G_1|\le|G_2|$) as input. In the algorithm, $\au_\leqslant(A_1,A_2)$ denotes the use of a function that outputs a $\leqslant$-common generalization on the atomic level for atoms $A_1$ and $A_2$ with respect to relation $\leqslant$. In our development we will call such functions \textit{anti-unification operators}. 
As stated in the following observation, such operators exist for our relations.

\begin{algorithm}[hbtp]
	\caption{Computing a lcg $G$ for goals $G_1$ and $G_2$ with generalization relation $\leqslant$} 
	\label{algo-rel-1-lcg}
	\begin{algorithmic}
		\State $G = \{\}, R = \{\}$ 
		\For {each ($A_1 \in G_1$)}
		\For {each ($A_2 \in G_2\setminus R$)}
		\State $A_1'$ = $\au_\leqslant(A_1, A_2)$
		\If{$A_1' \neq \bot$}				
		\State $G \gets G\cup A_1'$
		\State $R \gets R\cup A_2$
		\State \textbf{break} out of the inner loop
		\EndIf
		\EndFor
		\EndFor
		\State \textbf{return} $G$
	\end{algorithmic}
\end{algorithm}

\begin{lemma}\label{lemma-au-op}
	There exist polynomial anti-unification operators to compute the $\leqslant$-lcg and/or the $\leqslant$-msg of two atoms. In particular for two atoms $A_1$ and $A_2$, there exist (1) an operator $\au_\sqsubseteq(A_1,A_2)$ computing a $\sqsubseteq$-lcg for ${A_1}$ and ${A_2}$ in $\mathcal{O}(n)$ with $n$ the arity of $A_1$; (2) an operator $\au_\preceq(A_1,A_2)$ computing a $\preceq$-lcg in $\mathcal{O}(m)$ with $m$ the maximum number of function applications in the argument terms of the atom $A_1$; (3) an operator $\dau_\sqsubseteq(A_1,A_2)$ computing a $\sqsubseteq$-msg with a complexity that is linear in the number of terms appearing in $A_1$.
\end{lemma}

Algorithm~\ref{algo-rel-1-lcg} merely applies a given anti-unification operator to pairs of atoms and keeps the results (if not $\bot$) in the generalization under construction, leading to the conclusion:
\begin{theorem}\label{thm-ausqsubseteq}
	Given two goals $G_1$ and $G_2$, Algorithm~\ref{algo-rel-1-lcg} can compute (1) a $\sqsubseteq$-lcg in $\mathcal{O}(|G_1|\cdot |G_2|\cdot N)$ with $N$ the maximum arity of the predicate symbols occurring in $G_1$ and $G_2$; (2) a $\preceq$-lcg in $\mathcal{O}(|G_1|\cdot |G_2|\cdot N)$ with $M = \underset{A\in G_1}{\max}\{|\allterms(A)|\}$. 
\end{theorem}

Note that although Algorithm~\ref{algo-rel-1-lcg} is able to find a $\sqsubseteq$-lcg for two goals $G_1$ and $G_2$, it can produce different lcg's depending on the order in which the atoms of $G_1$ and $G_2$ are considered. 
%
Although the $\preceq$-lcg computed by Algorithm~\ref{algo-rel-1-lcg} is necessarily a $\preceq$-msg (according to Proposition~\ref{prop-msg-lcg-preceq}), the same observation does not hold when the underlying relation is $\sqsubseteq$ and the anti-unification operator is adapted accordingly. The fact that Algorithm~\ref{algo-rel-1-lcg} can miss out on a $\sqsubseteq$-msg is due to the algorithm itself not trying to match those pairs of atoms $(A_1, A_2)$ that share as much structure as possible. Therefore, finding a $\sqsubseteq$-lcg with maximal $\terms$-value (i.e. a $\sqsubseteq$-msg) can be seen as an optimization problem.

%
%

Indeed, applying Algorithm~\ref{algo-rel-1-lcg} as-is does not guarantee that the matched atoms from $G_1$ and $G_2$ are chosen in a way that optimizes the output's $\terms$-value. The algorithm should be adapted in such a way that first, the anti-unification of $A_1$ and $A_2$ is computed for \textit{all} $A_1\in G_1$ and $A_2\in G_2$; then, there must be a selection of pairs of atoms so that the resulting generalization has a maximized $\terms$-value. This is similar to the well-known assignment problem, and can consequently be solved by existing maximization matching algorithms~\cite{CATTRYSSE1992260}. Indeed, with $G_1$ and $G_2$ the goals at hand, our problem can be characterized by drawing a weighted bipartite graph with as left vertexes the atoms of $G_1$ and as right vertexes the atoms of $G_2$. When considering as granted an operator $\dau\footnote{For \textit{deep anti-unification}}_\sqsubseteq$ computing a $\sqsubseteq$-msg for two atoms, an edge between two vertexes $A_1$ and $A_2$ has an associated weight $w$ indicating the potential benefit (in number of terms and predicate symbols) of anti-unifying $A_1$ and $A_2$, formally defined as 
\[w(A_1,A_2)=\left\{\begin{array}{ll}
-1 & \mbox{if } \dau_\sqsubseteq(A_1, A_2) = \bot
\\ |\terms(\dau_\sqsubseteq(A_1,A_2))| & \mbox{otherwise}
\end{array}\right.\]

Since all edges are labeled by a measurement of their $\terms$-value, the maximum weight matching (MWM) in the bipartite graph will give the selection of pairs of atoms that, once properly anti-unified, keep the maximal structure in the generalization. Observe that by giving negative scores to atom couples that do not anti-unify, we prevent these couples from playing any part in the computed generalization. 

\begin{example}\label{example-mwm}
	Let us consider the goals $G_1 = \{p(X, t(4)), r(u(5,s(Y)),8), r(u(8,Z),5)\}$ and $G_2 = \{p(A), r(u(8,s(3)),5)\}$. The corresponding assignment problem is shown in Figure~\ref{fig:mwm}. The MWM consists of the sole edge $(r(u(8, Z), 5), r(u(8, s(3)), 5)$, so that the resulting generalization for this simple example is $G = \{r(u(8, \Phi(Z, s(3))), 5)\}$.
\end{example}

\begin{figure}
	\begin{tikzpicture}[thick,
	fsnode/.style={draw, circle},
	ssnode/.style={draw, circle},
	every fit/.style={ellipse,text width=3cm},
	-,shorten >= 3pt,shorten <= 3pt,
	]
	
	\begin{scope}[start chain=going below,node distance=6mm]
	\node[fsnode,on chain] (g1) [label=left: {$p(X, t(4))$}] {};
	\node[fsnode,on chain] (g2) [label=left: {$r(u(5,s(Y)),8)$}] {};
	\node[fsnode,on chain] (g3) [label=left: {$r(u(8,Z),5)$}] {};
	\end{scope}
	
	\begin{scope}[xshift=3.6cm,start chain=going below,node distance=10mm]
	\node[ssnode,on chain] (h1) [label=right: {$p(A)$}] {};
	\node[ssnode,on chain] (h2) [label=right: {$r(u(8,s(3)),5)$}] {};
	\end{scope}
	
	\node [black, fit=(g1) (g3),label=above:{$G_1$}] {};
	\node [black, fit=(h1) (h2),label=above:{$G_2$}] {};
	
	\path (g1) edge[bend left] node [fill=white] {$-1$} (h1);
	\path (g1) edge node [fill=white] {$-1$} (h2);
	
	\path (g2) edge[bend left] node [fill=white] {$-1$} (h1);
	\path (g2) edge node [fill=white] {$3$} (h2);
	
	\path (g3) edge[bend right] node [fill=white] {$-1$} (h1);
	\path (g3) edge[bend right] node [fill=white] {$4$} (h2);
	]	
	\end{tikzpicture}
	\caption{The bipartite graph for the assignment problem from Example~\ref{example-mwm}}
	\label{fig:mwm}
	
\end{figure}
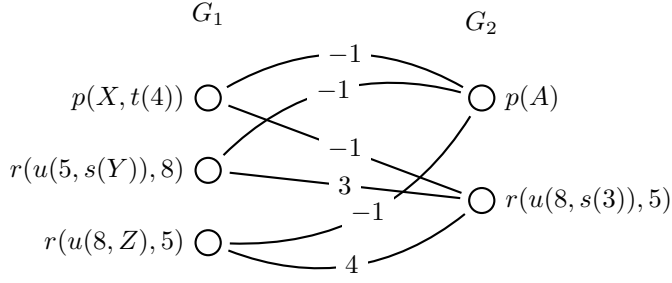


\begin{theorem}\label{thm-sqsubseteq-msg}
	Let $G_1$ and $G_2$ be goals and $N = \underset{A\in G_1}{\max}\{|\allterms(A)|\}$. Then a $\sqsubseteq$-msg of $G_1$ and $G_2$ can be computed in $\mathcal{O}\big(|G_1|\cdot |G_2|\cdot N +max(|G_1|,|G_2|)^3\big)$.
\end{theorem}
%

Note that the process described above finds \textit{a} $\sqsubseteq$-msg but there is no guarantee regarding \textit{which} $\sqsubseteq$-msg is found: as previously observed, the maximal $\terms$-value can sometimes be reached through different atomic structures. Another inconstant parameter from one msg to the other is the number of \textit{different} variables that are introduced in the generalization process. In fact, both aspects can sometimes be related, for example when minimizing the number of variables leads to the choice of a certain msg structure over another. A $\sqsubseteq$-most specific generalization that has \textit{as few} different variables as possible is often seen as an even more specific generalization; the computation of such a msg is the main topic of the following section.
	
	\section{Dataflow Optimization}\label{section-relation-2}

Relations $\sqsubseteq$ and $\preceq$ are defined over substitutions that do not necessarily need to be \textit{injective}. Indeed, a single term occurring multiple times in one of the goals can potentially be generalized by two (or more) different variables. Therefore, some most specific generalizations may contain more different variables than others depending on the underlying variabilization process. 
Among two common generalizations of the same pair of goals, the common generalization that has more variables than the other can be considered \textit{less specific} as some information -- namely the fact that two or more values, possibly in different atoms, are equal -- has been abstracted by introducing different variables. In what follows, we will call the search of a common generalization with as few different variables as possible \textit{dataflow optimization}. The following example illustrates the concept over the finite domain from~\cite{clpbfd}. 

\begin{example}\label{ex:dataflow-bools}
	Consider the domain of Booleans $\mathbb{B} = \{true, false\}$ as well as the following goals: $G_1 = \{=\!\!(X, or(Y, Z)), =\!\!(V, and(Y, Z))\}$ and $G_2 = \{=\!\!(B, or(C, D)), =\!\!(A, and(C, D)), =\!\!(E, and(F, G))\}$. 
	Note that in $G_1$ the \textit{or} and \textit{and} operations are evaluated on the same values, represented by the multiple occurrences of the variables $Y$ and $Z$. In $G_2$ the \textit{or} and the \textit{and} operation from the second atom exhibit this very same behavior (represented by the variables $C$ and $D$), whereas the third atom represent an \textit{and} operation on different values. 
	Computing a $\preceq$-msg (and in this example, a $\sqsubseteq$-msg) for $G_1$ and $G_2$ can lead to two different generalizations, namely
	\[
	\begin{array}{lll}
	G & = & \{=\!\!(\Phi(X,B), or(\Phi(Y,C), \Phi(Z,D))), =\!\!(\Phi(V, E), and(\Phi(Y, F), \Phi(Z,G)))\}\\ 
	G' & = & \{=\!\!(\Phi(X,B), or(\Phi(Y,C), \Phi(Z,D))), =\!\!(\Phi(V, A), and(\Phi(Y, C), \Phi(Z,D)))\}
	\end{array}
	\] 
	Clearly, both generalizations are correct msg's, but the fact that all the variables in $G$ only occur once merely denotes that there exist six variables that together can make $G$ true. The repetition of $Y$ and $Z$ in $G_1$ as well as their connection with $C$ and $D$ is a lost information, abstracted by the anti-unification process. On the other hand, $G'$ by harboring less different variables introduces less variable abstraction, effectively depicting some dataflow logic that is common to $G_1$ and $G_2$, through the occurrence of $\Phi(Y,C)$ and $\Phi(Z,D)$ in both its atoms. On that level, $G'$ can be considered less general than $G$. 
\end{example}

Dataflow optimization thus formally boils down to finding, among a group of common generalizations for two goals $G_1$ and $G_2$, a goal $G$ such that $|\vars(G)|$ is minimal. In Example~\ref{ex:dataflow-bools}, we were interested in finding, among all possible msg's of $G_1$ and $G_2$, one that harbors a minimal number of variables; it makes sense, since abstracting one Boolean value with two different variables can be too liberal, depending on the applications. In that case of dataflow optimization, where the target goal must be a msg (i.e. when both structure and dataflow must be optimized), the dataflow problem is NP-complete. The same is true for lcg's. In order to show this formally, we consider a formulation in terms of decision problems.

\begin{theorem}\label{thm-dataflow-np-complete}
	Let MSG-MIN (resp. LCG-MIN) denote the following decision problem: "Given goals $G_1$, $G_2$ and a constant $p\in \mathbb{N}_0$, does there exist a $\leqslant$-msg (resp. $\leqslant$-lcg) of $G_1$ and $G_2$ that has less than $p$ different variables?". MSG-MIN and LCG-MIN are NP-complete.
\end{theorem}
Now instead of looking to \textit{minimize} the number of different variables in the computed generalization $G$, one could be interested in \textit{forcing} to preserve all the dataflow implied in the generalized goals, not allowing to abstract away the links that appear in the goals' terms. Intuitively, this can be done by forbidding any term from one of the input goals to have more than one "corresponding term" in the other input goal. In other words, the dataflow is considered entirely preserved if the underlying variabilization function $\Phi$ doesn't associate any term with two or more different terms at the same time. Formally, this amounts to using an \textit{injective version} of our generalization relations. We say that a generalization relation is injective if its definition only holds for injective substitutions. For a common generalization $G$ of goals $G_1$ and $G_2$ and for some function $\Phi$ associating fresh variable names to couples of variables, this implies when using an anti-unification algorithm (e.g. Algorithm~\ref{algo-rel-1-lcg}) that for any two different variables $\Phi(T_1, T_2)$ and $\Phi(T_3, T_4)$ appearing in $G$, it holds that $T_1 \neq T_3 \neq T_2 \neq T_4\neq T_1$. We will denote by $\sqsubseteq^\iota$ (resp. $\preceq^\iota$) the versions of $\sqsubseteq$ (resp. $\preceq$) that exhibit this property.

\begin{example}
    Consider the injective relation $\preceq^\iota$ as well as the goals 
    	$G_1 = \{and(A,B), or(B,C), xor(C,A)\}$ and
    	$G_2 = \{and(X,Z), or(Y,X), xor(Z,Y)\}$.
    The only common generalizations are $\emptyset$, $\{and(\Phi(A,X),\Phi(B,Z))\}, \{or(\Phi(B,Y), \Phi(C,X))\}$ and $\{xor(\Phi(C,Z), \Phi(A,Y))\}$. No common generalization of size larger than 1 exists, since (at least) one of the matching substitutions is not injective. For example, the goal $G = \{and(\Phi(A,X), \Phi(B,Z)), or(\Phi(B,Y), \Phi(C,X))\}$ is not a common generalization of $G_1$ and $G_2$, since (at least) one of the substitutions mapping this goal to $G_1$ or $G_2$ is not injective. Indeed, the substitution $[\Phi(A,X) \mapsto A, \Phi(B,Z) \mapsto B, \Phi(B,Y) \mapsto B, \Phi(C,X) \mapsto C]$ maps both $\Phi(B,Z)$ and $\Phi(B,Y)$ to $B$; this is sufficient to reach the conclusion that $G$ is not an injective generalization of $G_1$ and $G_2$. Note that in this case, the other potential substitution, i.e. the one mapping $G$ on $G_2$, is not injective either. 
\end{example}

The two following observations immediately result from the injective relations being more constrained versions of their non-injective counterparts. 

\begin{proposition}
	Relations $\sqsubseteq^\iota$ and $\preceq^\iota$ are quasi-orders. 
\end{proposition}

\begin{proposition}
	Let $G_1$ and $G_2$ be goals. If $G_1\sqsubseteq_\theta^\iota G_2$, then $G_1\sqsubseteq_\theta G_2$. If $G_1\preceq^\iota_\theta G_2$, then $G_1\preceq_\theta G_2$ and $G_1\sqsubseteq^\iota_\theta G_2$.
\end{proposition}

With an injective generalization relation, the computing of a msg is fundamentally dissociated from that of an lcg, as an msg is not necessarily a lcg due to the injectivity constraint. However, both situations are intractable. In order to show this formally, we define the following decision problem variant.

\begin{theorem}\label{thm-inj-np-complete}
	Let INJ denote the following decision problem: "Given an injective generalization relation $\leqslant^\iota$ along with goals $G_1$ and $G_2$ such that $|G_1|\le|G_2|$, verify whether there exists an ad hoc injective substitution $\theta$ such that $G_1\theta\subseteq G_2$". INJ is NP-complete.
\end{theorem}

INJ is basically the verification of whether a goal $G_1$ can be adequately mapped onto (a subset of) another goal $G_2$. If there exists a substitution $\theta$ (resp. a renaming $\rho$) making this possible, then $G_1$ is a $\sqsubseteq^\iota$- (resp. $\preceq^\iota$-)largest \textit{and} most specific generalization of $G_1$ and $G_2$, since no larger nor structurally more specific goal than $G_1$ can exist for this specific situation.

Due to the inherent intractability of injective relations, it is sometimes preferable to make use of tractable abstractions rather than exact brute-force algorithms, especially if a quick and approximate (though entirely dataflow-preserving) anti-unification result suffices for the application at hand. In the next section, we give such an efficient -- yet highly accurate~-- abstraction for the computation of $\preceq^\iota$-lcg's.

	\section{The $k$-swap Stability Abstraction}\label{section-relation-3}
In what follows, we introduce an abstraction for the largest common generalization with respect to $\preceq^\iota$ that can be computed in polynomial time. The abstraction was already introduced in~\cite{gen} but no formal proof of its complexity was given. The abstraction is based on the \textit{$k$-swap stability} property, which is in turn defined in terms of \textit{pairing generalizations}. 

\begin{definition}
	Let $G_1$ and $G_2$ be two renamed apart goals and $G$ be a $\preceq^\iota$-common generalization of $G_1$ and $G_2$ such that $G\subseteq G_1$. Let $\rho$ be any renaming such that $G\rho\subseteq G_2$. The \emph{pairing generalization} of $G$, denoted $\pi(G)$, is the set of pairs $(A_1, A_2) \in G_1\times G_2$ such that $\forall(A_1, A_2)\in\pi(G):A_1\rho = A_2$. 
\end{definition}

\begin{example}
	Considering the goals $G_1 = \{p(A), p(B), q(A)\}$ and $G_2 = \{p(X), q(Y)\}$, it is easy to see that $G = \{p(\Phi(B,X)), q(\Phi(A,Y))\}$ is a $\preceq^\iota$-common generalization of them. The corresponding pairing generalization is $\pi(G) = \{(p(B), p(X)), (q(A), q(Y))\}$.
\end{example} 

The notion of a pairing generalization renders thus explicit the corresponding atoms from the generalized goals that contribute to the generalization.
As a slight abuse of language, given a pairing generalization $\pi$ of some generalization $G$ for goals $G_1$ and $G_2$, we will simply say that $\pi$ is a \textit{pairing for $G_1$ and $G_2$}. 
Pairings can be used to express a notion of goal \textit{stability} in the following sense.

\begin{definition}\label{def-k-swap-stable}
	Let $G_1$ and $G_2$ be two renamed apart goals and $G$ be a $\preceq^\iota$-common generalization of $G_1$ and $G_2$ such that $G\subseteq G_1$. $G$ is \emph{k-swap stable} if and only if there does not exist some generalizations $\hat{G}$ and $G'$ of $G_1$ and $G_2$ such that $\hat{G}\supset G'$ and $|\pi(G)\cap\pi(G')|\ge |\pi(G)|-k$ for some $k\in\mathtt{N}$.
\end{definition}

Intuitively, a generalization $G$ is $k$-swap stable if it is impossible to transform $G$ into a larger generalization $\hat{G}$ in spite of ``swapping'' at most $k$ pairs in $\pi(G)$.
This stability notion gives a characterization of the quality of a computed generalization. If a generalization is 0-swap stable (the weakest characterization), it cannot be extended by adding another atom but this guarantees in no way that a larger generalization could not be found. If a generalization $G$ is $k$-swap stable (for $k>0$), it means that even if we exchange up to $k$ pairs in $\pi(G)$ by others, the generalization cannot be extended into a larger one. Consequently, if a generalization is $k$-swap stable for $k$ the number of atoms in the smallest of the two goals (denoted by $\infty$-swap stable), it means that the computed generalization is a largest common generalization. 
Operationally, when naively searching for a lcg by backtracking, the fact that a computed generalization is $k$-swap stable means that one should backtrack by \textit{more} than $k$ choice points in order have a chance of finding a larger generalization. 


\begin{example}\label{ex:k-swap-stable}
   	Consider the goals $G_1 = \{add(X,Y,Z), even(X), odd(Z), p(Z)\}$ and $
G_2 = \{add(A,B,C), add(C,B,A), even(C), odd(A), p(C)\}$. 
   	 $\pi_1= \{(add(X,Y,Z), add(A,B,C))\}$ is not $0$-swap stable. Indeed, we can enlarge $\pi_1$ by adding $(p(Z),p(C))$, in order to obtain 
   	$\pi_2 = \{(add(X,Y,Z),add(A,B,C)), (p(Z),p(C))\}.$
   	Note that $\pi_2$ is $0$-swap stable, it is impossible to add another pair to $\pi_2$ and still obtain a common generalization. 
   	It is also $1$-swap stable, seeing that replacing (or removing) one of the pairs doesn't lead to a pairing readily extensible to a pairing of size strictly greater than~$2$. However, $\pi_2$ is not $2$-swap stable. Indeed, replacing the pair $(add(X,Y,Z),add(A,B,C))$ by the pair $(add(X,Y,Z), add(C,B,A))$ in $\pi_2$ and removing the now incompatible pair $(prime(Z),prime(C))$ (i.e. choosing the renaming $[X\mapsto C, Y\mapsto B, Z\mapsto A]$ instead of $[X\mapsto A, Y\mapsto B, Z\mapsto C]$) gives rise to $\pi_2' = \{(add(X,Y,Z),add(C,B,A))$, which can readily be extended into
   	   		$\pi_3 = \{(add(X,Y,Z),add(C,B,A)), (even(X),even(C)), (odd(Z),odd(A))\}$
   	which is a pairing of size 3. The latter being $\infty$-swap stable, it represents a $\preceq^\iota$-lcg, namely
   		$\hat{G} = \{add(\Phi(X,C), \Phi(Y,B), \Phi(Z,A)), even(\Phi(X,C)), odd(\Phi(Z,A))\}$
\end{example}
An algorithm has been introduced in~\cite{gen} that builds up a $k$-swap stable generalization using the process suggested in Example~\ref{ex:k-swap-stable}. Its practical performance has been assessed on different test cases. The tests indicate that the $k$-swap stability property represents a well-suited approximation of the concept of $\preceq^\iota$-lcg. Indeed, in all test cases the size of the $k$-swap stable generalization was at least $90\%$ of the size of an lcg for the same anti-unification problem, while the computational time was radically reduced -- especially as the size of the input goals grows\footnote{For example, with $k$ fixed to 4, anti-unifying goals harboring 15 to 22 atoms, each of arity between 1 and 3, comes on average down from more than 7 minutes (using bruteforce) to 272 milliseconds (using the algorithms presented in this section), while the size of the computed generalization is on average $95\%$ of the size of a lcg. More detailed test results are exposed in~\cite{gen}.}. However, in~\cite{gen} only pragmatical aspects have been explored; the theoretical foundations of the $k$-swap technique were not detailed, and  no actual time complexity upper bound has been demonstrated. We fill this gap in the remainder of this section. First, we introduce the algorithm, then we formally prove that its time complexity is polynomially bounded. Before introducing the algorithm, which is essentially composed of two sub-algorithms, we give some notations that will facilitate their formulation. First, we define an operator that allows to combine two pairings into a single pairing.
	
\begin{definition}
	Let $G_1$ and $G_2$ be two renamed apart goals. The \emph{enforcement operator} is defined as the function $\enforce: (G_1\times G_2)^2 \mapsto (G_1\times G_2)$ such that for two pairing generalizations $\pi$ and $\pi'$ for $G_1$ and $G_2$, $\pi \enforce \pi' = \pi' \cup M$ where $M$ is the largest subset of $\pi$ such that $\pi'\cup M$ represents a $\preceq^\iota$-common generalization of $G_1$ and $G_2$. 
\end{definition}
	
In other words, $\pi\enforce\pi'$ is the mapping obtained from $\pi\cup\pi'$ by eliminating those pairs of atoms $(A,A')$ from $\pi$ that are \textit{incompatible} with some $(B,B')\in\pi'$ either because they concern the same atom(s) or because the involved renamings cannot be combined into a single injective renaming. 
	
\begin{example}
	Consider 
	$\pi = \{(p(X, Y), p(A, B)), (q(X), q(A))\}$ as a pairing for two goals $G_1$ and $G_2$. Suppose $\pi' = \{(r(Y), r(C))\}$ is also a pairing for $G_1$ and $G_2$. Enforcing $\pi'$ into $\pi$ gives $\pi\enforce\pi' = \{(q(X), q(A)), (r(Y), r(C))\}$. Indeed, this can be seen as forcing $Y$ to be mapped on $C$; therefore the resulting pairing generalization can no longer contain $(p(X, Y), p(A, B))$ as the latter maps $Y$ on $B$. 
\end{example}

For $\pi_1$ and $\pi_2$ pairings we will also denote by $\compatible_{\pi_1}(\pi_2)$ the subset of $\pi_2$ of which each element can be added to $\pi_1$ such that the result is a pairing (i.e. there is no injectivity conflict in the associated renaming). Finally, we use $\gen(G_1, G_2)$ to represent those atoms from $G_1$ and $G_2$ that are variants of each other, formally $\gen(G_1,G_2)=\{(A,A')\:|\:A\in G_1, A'\in G_2\mbox{ and }A\rho=A'\mbox{ for some renaming }\rho\}$. 
The first algorithm is depicted in Algorithm~\ref{alg:kswap}. The algorithm represents the construction of a $k$-swap stable generalization of goals $G_1$ and $G_2$. At each round, the process tries to transform the current generalization $\pi$ (which initially is empty) into a larger generalization by forcing a new pair of atoms $(A,A')$ from $\gen(G_1,G_2)$ in $\pi$, which is only accepted if doing so requires to swap no more than $k$ elements in $\pi$. More precisely, the algorithm selects a subset of $\pi$ (namely $\pi_s$) that can be swapped with a subset $\pi_c$ of the remaining mappings from $\gen(G_1,G_2)\setminus\pi$ such that the result of replacing $\pi_s$ by $\pi_c$ in $\pi$ and adding $(A,A')$ constitutes a pairing. Note how condition 1 in the algorithm expresses that $\pi_s$ must include at least those elements from $\pi$ that are not compatible with $(A,A')$.
The search continues until no such $(A,A')$ can be added.

\begin{algorithm}[hbtp]
	\caption{Computing a $k$-swap stable generalization $G$ for goals $G_1$ and $G_2$}
	\label{alg:kswap}
	\begin{algorithmic}
		\State $\pi\gets\emptyset$
		\Repeat
			\State $found\gets false$
			\ForAll{$(A,A')$ in $\gen(G_1, G_2)\setminus \pi$}
				\State select $\pi_s \subseteq \pi$ and $\pi_c\subseteq \gen(G_1, G_2)\setminus(\pi \cup\{(A,A')\})$ such that:
				\State \hspace*{3ex}(1) $\pi_s\supseteq \pi\setminus \pi\enforce\{(A,A')\}$
				\State \hspace*{3ex}(2) $|\pi_s| \le k$
				\State \hspace*{3ex}(3) $|\pi_c| = |\pi_s|$
				\State \hspace*{3ex}(4) $\pi\setminus \pi_s\cup \pi_c\cup\{A,A'\}$ is a pairing generalization of $G_1$ and $G_2$
				
				\If{such $\pi_c$ and $\pi_s$ are found}
					\State $\pi\gets \pi\setminus \pi_s\cup \pi_c \cup \{(A,A')\}$
					\State $found\gets true$
					\State \textbf{break} out of the \textbf{for} loop
				\EndIf
			\EndFor
		\Until $\neg found$
		\State $G \gets \dom(\pi)$
	\end{algorithmic}
\end{algorithm}	

The main operation of Algorithm~\ref{alg:kswap}, namely the selection of $\pi_s$ and $\pi_c$, is detailed in Algorithm~\ref{alg:selection} which aims to select the parts of the pairings to be swapped in order to enlarge the resulting pairing under construction ($\pi$) by the couple $(A,A')$.
To that purpose $\pi_s$ is initialized with the part of $\pi$ that is incompatible with the pair of atoms $(A,A')$ that we wish to enforce into the generalization. Its replacement mapping $\pi_c$ is initially empty and the algorithm subsequently searches to construct a sufficiently large $\pi_c$ (the inner while loop). During this search, $S$ represents the set of candidates, i.e. couples from $\gen(G_1,G_2)$ that are not (yet) associated to the generalization. In order to explore different possibilities with backtracking, the while loop manipulates a stack $GS$ that records alternatives for $\pi_c$ with the corresponding set $S$ for further exploration. 
\begin{algorithm}[hbtp]
\caption{Selecting $\pi_s$ and $\pi_c$ for a given $(A,A')$}
\label{alg:selection}
	\begin{algorithmic}
		\State $GS \gets \{\}$,  $BS \gets \{\}, \pi_c \gets\{\}$
		\State $\pi_s \gets \pi\setminus\pi\enforce\{(A,A')\}$
		\State $S \gets \gen(G_1, G_2) \setminus\pi\enforce\{(A,A')\}$
		\While{$|\pi_c| < |\pi_s| \mbox{ and } |\pi_s|\le k$}
			\While{$|\pi_c| < |\pi_s| \mbox{ and } \neg (\compatible_{\pi\setminus\pi_s\cup\pi_c}(S) = \{\} \mbox{ and } GS = \{\})$}
				\ForAll{$p$ in $\compatible_{\pi\setminus\pi_s\cup\pi_c}(S)$}	
					\State $push(GS, (\pi_c\cup p, S\setminus\{p\}))$
				\EndFor
				\State $(\pi_c, S)\gets pop(GS)$
			\EndWhile
			\If{$|\pi_c|<|\pi_s|$}
				\ForAll{$p$ in $\pi\setminus\pi_s$}
					\State $enter(BS, \pi_s\cup\{p\})$
				\EndFor
				\If{$BS\neq\{\}$}
					\State $\pi_s\gets exit(BS)$
					\State $\pi_c\gets\{\}$
					\State $S\gets \gen(G_1, G_2)\setminus(\pi\cup \{(A,A')\})$
				\Else
					\State \textbf{return} $\bot$
				\EndIf
			\EndIf
		\EndWhile
		\If{$|\pi_c|=|\pi_s|$}
			\State \textbf{return} $\pi_s, \pi_c$
		\EndIf
			\State r\textbf{return} $\bot$
	\end{algorithmic}
\end{algorithm}

If the search for $\pi_c$ was without a satisfying result (i.e. no $\pi_c$ is found equal in size to $\pi_s$), the algorithm continues by removing another couple from $\pi$ (thereby effectively enlarging $\pi_s$). The rationale behind this action is that there might be a couple in $\pi$ that is ``blocking'' the couples in $S$ from addition to $\pi$. In order to achieve the removal of such potentially blocking couples, an arbitrary couple from $\pi\setminus\pi_s$ is selected, and alternatives are recorded in a queue ($BS$). Note the use of a queue (and its associated operations \textit{enter} and \textit{exit}) as opposed to the stack $GS$.
The process is repeated until either $|\pi_c|=|\pi_s|$ in what case we have found a suitable $k$-swap, or until $|\pi_s|>k$ in what case we have not, and the algorithm returns $\bot$. 

While the algorithms have been proven to correctly compute a $k$-swap stable generalization~\cite{gen}, no result on their complexity has yet been formally established. 

\begin{theorem}\label{thm-k-swap-stable}
	For a given and constant value of $k$, the combination of Algorithms~\ref{alg:kswap} and~\ref{alg:selection} computes a $k$-swap stable common generalization of input goals $G_1$ and $G_2$ in polynomial time $\mathcal{O}((\alpha M)^{k+1})$, with $0 \le M \le |gen(G_1, G_2)|$ and $0 \le \alpha \le \textit{min}(|G_1|,|G_2|)$.
\end{theorem}
\begin{proof}
	In order to search for a suited $\pi_c$ to be swapped with a certain $\pi_s$, Algorithm~\ref{alg:selection} must try to add $|\pi_s|$ couples to $\pi\setminus\pi_s$ among the couples in $S$ that are compatible with it. To simplify notation, let $i = |\pi_s|$ and $n = |\compatible_{\pi\setminus\pi_s\cup\pi_c}(S)|$. Note that at any moment $i \le k$. The attempt of Algorithm~\ref{alg:selection} to find $\pi_c$ is essentially a search of a combination of $i$ couples among $n$; that is $\binom{n}{i}$ possibilities to explore. We have $\binom{n}{i} = \frac{n!}{i!(n-i)!}$ which reduces to a polynomial of degree $n^i$:
	\begin{gather*}
	\begin{array}{lll}
	\frac{n!}{i!(n-i)!} &=& \frac{n\cdot (n-1)\dots\cdot (n-(i+1) \cdot (n-i)\cdot (n-(i-1))\cdot\dots\cdot 1}{i!\cdot(n-i)\cdot (n-(i-1))\cdot\dots\cdot 1}
	= \frac{n\cdot (n-1)\dots\cdot (n-(i+1))}{i!}
	\approx \mathcal{O}(n^i)
	\end{array}
	\end{gather*}
	
	If no suiting $\pi_c$ is found during such a search, 
	then $\pi_s$ gets enlarged, having its size $m$ increased by (at least) one unit.
	In the worst case, the size $i$ of $\pi_s$ is, at the start of Algorithm~\ref{alg:selection}, equal to $1$. It then gets incremented by one, until it reaches $k$ (each time more atoms from $\pi$ being considered to be part of $\pi_s$). 
	Let $p$ denote the size of the pairing $\pi$ under construction, that is $p = |\pi|$. As $k$ is constant, if  backtracking is exhaustive there are $\sum\limits_{i=1}^{k} \binom{p}{i}$
	possibilities for $\pi_s$ pairings that are explored this way. Each of these $\pi_s$ pairings leads to the search for a corresponding $\pi_c$ pairing. As such, the overall search carried out by Algorithm~\ref{alg:selection} takes a number of iterations that is in the worst case represented by
	\[
	\begin{array}{lllll}
		\sum\limits_{i=1}^{k} \binom{p}{i}\cdot\binom{n}{i} & \approx & \sum\limits_{i=1}^{k} \mathcal{O}(p^i)\cdot\mathcal{O}(n^i) & \approx & \mathcal{O}((p\cdot n)^k)
	\end{array}
	\]
Given that $n$ is bound by the number of compatible couples of atoms from $G_1\times G_2$, we will denote the worst-case time complexity of Algorithm~\ref{alg:selection} by $\mathcal{O}((p\cdot M)^k)$ with $M \le |\gen(G_1,G_2)|$ and $p$ the length of the pairing under construction $\pi$.
		

	
Turning our attention to Algorithm~\ref{alg:kswap} it is clear that the size of pairing $\pi$ is incremented by $1$ in each iteration of the \textit{repeat}-loop, since $found$ must be true for a new iteration to occur. As such, in the worst-case scenario there can be as many iterations as there are atoms in the smallest goal amongst $G_1$ and $G_2$, seeing that a generalization size cannot exceed that of the goals it generalizes. We will denote this number by $\alpha = \min(|G_1|, |G_2|)$.
As for the inner loop of Algorithm~\ref{alg:kswap}, it can browse through up to $|\gen(G_1, G_2)| - p$ candidates for choosing the couple $(A,A')$ that will be enforced in the pairing $\pi$. This gives us at most $
	\begin{array}{lll}
	\sum\limits_{p=1}^{\alpha}(|\gen(G_1, G_2)| - p) &\approx &\sum\limits_{p=1}^{\alpha}\mathcal{O}(M - p)
	\end{array}
	$ 
	iterations of Algorithm~\ref{alg:kswap}.
	Algorithm~\ref{alg:selection} being called at each inner loop iteration of Algorithm~\ref{alg:kswap}, we can represent the time complexity of the combined algorithms by
	$\sum\limits_{p=1}^{\alpha}\mathcal{O}(M-p) \cdot \mathcal{O}((p\cdot M)^k) \approx \sum\limits_{p=1}^{\alpha} \left((M-p)\cdot p^k\cdot M^k\right)$
	which can be rewritten as $
	M^{k+1}\cdot \left(\sum\limits_{p=1}^{\alpha} p^k\right) - M^k \cdot \left(\sum\limits_{p=1}^{\alpha} p^{k+1}\right)$.
	
	Since $\sum\limits_{p=1}^{\alpha} p^k \approx \mathcal{O}(\alpha^{k+1})$ and $\sum\limits_{p=1}^{\alpha} p^{k+1} \approx \mathcal{O}(\alpha^{k+2})$, we can conclude the total complexity to be of the order $\mathcal{O}((\alpha\cdot M)^{k+1}) - \mathcal{O}(\alpha^{k+2}\cdot M^k)$ which proves the result.
\end{proof}

Whenever there is a need to compute numerous anti-unifications of unordered goals with limited time resources, the $k$-swap stability abstraction allows to keep the search space tractable while outputting goals that are, on average, close in size to that of a lcg. Such situations can e.g. arise in static analysis techniques for large Horn clause programs, such as the assessment of structural similarity between algorithms expressed in CLP~\cite{clones}. 

	
	\section{Conclusions and Future Work}\label{section-conclusion}
In this work, we have systematically studied different key notions and results concerning anti-unification of unordered goals, i.e. sets of atoms. We have defined different anti-unification operators and we have studied several desirable characteristics for a common generalization, namely optimal cardinality (lcg), highest $\tau$-value (msg) and variable dataflow optimizations. For each case we have provided detailed worst-case time complexity results and proofs. An interesting case arises when one wants to minimize the number of generalization variables or constrain the generalization relations so as they are built on injective substitutions. In both cases, computing a relevant generalization becomes an NP-complete problem, results that we have formally established.
In addition, we have proven that an interesting abstraction -- namely $k$-swap stability which was introduced in earlier work -- can be computed in polynomially bounded time, a result that was only conjectured in  earlier work. 

Our discussion of dataflow optimization in Section~\ref{section-relation-2} essentially corresponds to a reframing of what authors of related work sometimes call the \textit{merging} operation in rule-based anti-unification approaches as in~\cite{Baumgartner2017}. Indeed, if the "store" manipulated by these approaches contains two anti-unification problems with variables generalizing the same terms, then one can "merge" the two variables to produce their most specific generalization. If the merging is exhaustive, this technique results in a generalization with as few different variables as possible. In this work we isolated dataflow optimization from that specific use case and discussed it as an anti-unification problem in its own right.

While anti-unification of goals in logic programming is not in itself a new subject, to the best of our knowledge our work is the first systematic treatment of the problem in the case where the goals are not sequences but unordered sets. Our work is motivated by the need for a practical (i.e. tractable) generalization algorithm in this context. The current work provides the theoretical basis behind these abstractions, and our concept of $k$-swap stability is a first attempt that is worth exploring in work on clone detection such as~\cite{clones}. 

Other topics for further work include adapting the $k$-swap stable abstraction from the $\preceq^\iota$ relation to dealing with the $\sqsubseteq^\iota$ relation. 
A different yet related topic in need of further research is the question about what anti-unification relation is best suited for what applications. For example, in our own work centered around clone detection in Constraint Logic Programming, anti-unification is seen as a way to measure the distance amongst predicates in order to guide successive syntactic transformations. Which generalization relation is best suited to be applied at a given moment and whether this depends on the underlying constraint context remain open questions that we plan to investigate in the future. 

	\bibliographystyle{plainurl}
	\bibliography{main}
	
	\newpage	
	\appendix
	
	\section{Proof of Proposition~\ref{prop-quasi-order}}
\begin{proof}
	We will prove the result for relation $\sqsubseteq$, the proof for $\preceq$ being similar. We need to prove that $\sqsubseteq$ is reflexive and transitive. For reflexivity, it is obvious that since $G\subseteq G$ for any goal $G$, we have $G\sqsubseteq_\theta G$ for the empty substitution $\theta$. For transitivity, suppose that for goals $G_1$, $G_2$ and $G_3$, it holds that $G_1 \sqsubseteq_{\theta_1} G_2$ and $G_2 \sqsubseteq_{\theta_2} G_3$. Then by Definition~\ref{def-generalization}, there exist sets of atoms $\Delta_1$ and $\Delta_2$ such that $G_1\theta_1 \cup \Delta_1 = G_2$ and $G_2\theta_2\cup\Delta_2 = G_3$. In other words it holds that $(G_1\theta_1\cup\Delta_1)\theta_2\cup\Delta_2 = G_3$ or equivalently, $(G_1\theta_1)\theta_2 \cup (\Delta_1\theta_2 \cup\Delta_2) = G_3$. As the composition of two substitutions is a substitution, by defining $\theta_3 = \theta_2\circ\theta_1$ and $\Delta_3 = \Delta_1\theta_2 \cup\Delta_2$, we have $G_1\theta_3 \cup \Delta_3 = G_3$, so $G_1\sqsubseteq_{\theta_3} G_3$, which concludes the proof.
\end{proof}

	\section{Proof of Proposition~\ref{prop-msg-lcg}}
	
		First, observe the following property that holds for both relations, essentially stating that a common generalization that is not a lcg has a direct extension obtained by the addition of one atom. 
		
		\begin{proposition}\label{prop-lcg-extensible}
			Let $G_1, \dots, G_n$ and $G$ be goals such that $G$ is a $\leqslant$-common generalization, but not a $\leqslant$-lcg, of $\{G_1, \dots, G_n\}$. Then there exists an atom $A\notin G$ such that $G\cup\{A\}$ is a $\leqslant$-common generalization of $\{G_1, \dots, G_n\}$.
		\end{proposition}

		\begin{proof}
		
		Let us suppose the existence of some goal $G$, a $\leqslant$-common generalization that is not a $\leqslant$-lcg of $\{G_1, \dots, G_n\}$, and let us try and extend $G$ into a $\leqslant$-common generalization $G\cup\{A\}$ with $A\notin G$ an atom. As $G$ is not a lcg, there must exist another goal $G'$ being a $\leqslant$-lcg of $G_1$ and $G_2$ and obviously we have $|G'|>|G|$. As a consequence at this point there are three groups of atoms that can be identified: let us denote by $\hat{A}_1, \dots, \hat{A}_p$ the $p (\ge 0)$ atom(s) that are both in $G$ and in $G'$; by $A_1, \dots, A_m$ the $m (\ge 0)$ atom(s) that are part of $G$ but not of $G'$; and by $B_1, \dots, B_l$ the $l (\ge 1)$ atom(s) that are part of $G'$ but not of $G$. For an element $A$ of any of these sets, we denote by $A^1, \dots, A_n$ the atom in respectively $G_1, \dots, G_n$ whose anti-unification led to having $A$ as part of the generalizations.
		
		From the fact that $|G'|>|G|$ it follows that $l>m$. Now each $A_i (i \in 1..m)$ is such that $\exists  h\in 1..n : A_i^h\in \{B_i^h|i\in 1..l\}$: if not, it would be possible to add an atom generalizing $\{A_i^1, \dots, A_i^n\}$ (such as $A_i$) in $G'$ and get a larger generalization, which is impossible given that $G'$ is a lcg. Also note that for two atoms $B_i$ and $B_j (1\le i < j \le l)$, for $g, h \in 1..n : g\neq h$, if $B_i^g$ is anti-unifiable with an atom $B_j^h$ then $B_j^g$ is also anti-unifiable with $B_i^h$ (as it means that all four base atoms are a call to one and the same predicate (with relation $\sqsubseteq$) or have the exact same inner structure save for variables (with relation $\preceq$)), so that it is possible to switch the atoms $B_i^g$ and $B_j^h$, compute the anti-unification of $\{B_i^g, B_j^h)$ and $(B_j^g, B_i^h)$, and get an equally valid anti-unification. Thanks to this we can, where necessary, perform switches so as to rearrange the atoms $B_i (1\le i \le l)$ into $\{\hat{B}_i|i \in 1..l\}$ in such a way that $\{A_i|i\in 1..m\} \subset \{\hat{B}_i|i \in 1..l\}$ and for each atom $\hat{B}_i$, either $\hat{B}_i \in \{A_k|k\in 1..m\}$ or $\exists g, h \in 1..n : g\neq h \wedge \hat{B}_i^g\notin \{A_k^g|k\in 1..m\}\wedge \hat{B}_i^h\notin \{A_k^h|k\in 1..m\}$. We can now define a new generalization $\hat{G}$ defined as the union of these rearranged atoms and those that are common to $G$ and $G'$, i.e. $\hat{G} = \{\hat{A_i}|i\in 1..p\}\cup\{\hat{B}_i|i\in 1..l\}$. Since $|\hat{G}| = |G'|$ and $G\subset \hat{G}$, it suffices to add one of the atoms $A \in \hat{G}\setminus G$ to $G$ in order to obtain $G\cup\{A\}$, a $\leqslant$-common generalization of $\{G_1, \dots, G_n\}$ by construction.
	\end{proof}

		Next, we prove Proposition~\ref{prop-msg-lcg}.
		
	\begin{proof}
	We prove that any $\leqslant$-msg is a $\leqslant$-lcg by contradiction. Let us suppose that some goal $G$ is both a $\leqslant$-msg and not a $\leqslant$-lcg of the set of $\{G_1, \dots, G_n\}$. According to Proposition~\ref{prop-lcg-extensible} it must then be possible to select an atom $A \notin G$ such that $G\cup\{A\}$ is a $\leqslant$-common generalization of $\{G_1, \dots, G_n\}$. Since $A\notin G$ and any atom has a $\tau$-value of at least 1, it follows that $|\tau(G\cup\{A\})|>|\tau(G)|$. Consequently $G$ cannot be a $\leqslant$-most specific generalization of $G_1$ and $G_2$: a contradiction.
	
	As for the fact that any $\preceq$-lcg is a $\preceq$-msg, we prove this also by contradiction. Let $G$ represent a $\preceq$-lcg of the set of goals $\{G_1, \dots, G_n\}$ and let us suppose that $G$ is not a $\preceq$-msg. Then there must exist another goal that is a $\preceq$-msg of $\{G_1, \dots, G_n\}$, say $G'$, such that $|\tau(G')|>|\tau(G)|$ and, according to the first part of the proposition, $|G'|=|G|$. 
	Now, observe that for a set of atoms $\{A_1, \dots, A_n\}$ to be anti-unified with $\preceq$ into an atom $A$, necessarily all $A_i (1\le i\le n)$ must have the same $\terms$-value. Indeed relation $\preceq$ is defined upon renamings so that only variables (having a $\terms$-value of zero) are impacted by the generalization process. Therefore, the only possibility for the inequality $|\tau(G')|>|\tau(G)|$ to be true is that some atoms $B_1, \dots, B_n$ of respective goals $G_1,\dots,G_n$ appear in a generalized form (say $B$) in $G'$, while these atoms have not been generalized in $G$. This means that it is possible to add a (possibly renamed) version of $B$ in $G$ and obtain $G\cup\{B\}$, also a $\preceq$-common generalization and larger than $G$: a contradiction.
	\end{proof}

	\section{Detailed proof of Lemma~\ref{lemma-au-op}}
	\begin{proof}
The lemma will be shown correct by the definition of three anti-unification operators. A first anti-unification operator, based on $\sqsubseteq$ is the following. 

\begin{definition}
	\label{def-atoms-au}
	Given a variabilization function $\Phi$, let $\au^\Phi_\sqsubseteq$ (or simply $\au_\sqsubseteq$ if $\Phi$ is clear from the context) denote the anti-unification operator such that for any two atoms $A = a(t^A_1, \dots, t^A_n)$ and $B = b(t^B_1, \dots, t^B_m)$, it holds that \[\au^{\Phi}_\sqsubseteq(A,B)=\left\{\begin{array}{l}
		a\big(\Phi(t^A_1, t^B_1), \dots, \Phi(t^A_n, t^B_n)\big) \\ \qquad  \mbox{if } a = b \mbox{ and } n = m\\
		\bot \\
		\qquad \mbox{otherwise}\\
	\end{array}\right. \]
\end{definition} 

\begin{example}\label{ex-au-sq}
	In Table~\ref{table:sqsubseteq}, we show three atomic anti-unification results obtained by the application of $\au_\sqsubseteq^\Phi$ with $\Phi$ a given variabilization function. Note how in the first example, the predicates used in $A_1$ and $A_2$ differ (resp. $p/2$ and $p/3$), leading to an impossible anti-unification.
\end{example}

\begin{table*}
	\caption{Example results for $au_\sqsubseteq^\Phi$}
	\label{table:sqsubseteq}
	\centering
	\begin{tabular}{l|l|l}
		$\bm{A_1}$ & $\bm{A_2}$ & $\bm{\au_\sqsubseteq^\Phi(A_1, A_2)}$\\\hline 
		$p(X, 5, q(Y,4))$ & $p(W,t(Z))$ & $\bot$\\\hline 
		$p(r(X,3), t(5))$ & $p(W, t(Z))$ & $p(\Phi(r(X,3),W), \Phi(t(5), t(Z)))$\\\hline 
		$p(r(X,3), t(Y))$ & $p(r(W,3),t(Z))$ & $p(\Phi(r(X,3),r(W,3)), \Phi(t(Y),t(Z)))$ 
	\end{tabular} 
\end{table*}

Note that the anti-unification operator defined in Definition~\ref{def-atoms-au} differs from the traditional subsumption operator in the ordered case (i.e. when goals are ordered sequences of atoms). The difference comes from the fact that our goals being sets, all the possible couples of atoms have to be considered, whereas traditional subsumption must handle one atom at the time, making the anti-unification operator more straigtforward.     

Let us now introduce a second anti-unification operator that will allow to compute a $\preceq$-lcg. 
Since the result of this operator should be a $\preceq$-common generalization, the operator need only to anti-unify the \textit{variables} occurring at the corresponding positions in the atoms under investigation. 
The operator must thus go deeper into the term structure of the atoms than $\au_\sqsubseteq$ does, as it needs to only anti-unify those atoms that harbor the exact same structure at the level of their non-variable terms.
\begin{definition}
	\label{def-term-au-through-variables}
	Given some variabilization function $\Phi$, let $\au^\Phi_\preceq$ (or simply $\au_\preceq$ if $\Phi$ is clear from the context) denote the function such that for any two terms $T = t(t_1, \dots, t_n)$ and $U = u(u_1, \dots, u_m)$ it holds that
	\[\au^\Phi_\preceq(T,U)=\left\{\begin{array}{l}
		\Phi(T,U) 
		\\ \qquad \mbox{if } T\in\mathcal{V}\mbox{ and } U\in\mathcal{V}
		\\t\big(\au^\Phi_\preceq(t_1,u_1), \dots, \au^\Phi_\preceq(t_n, u_n)\big) 
		\\ \qquad \mbox{if } t = u \mbox{ and } n = m 
		\\ \qquad \mbox{and } \forall i \in 1..n: \au^\Phi_\preceq(t_i,u_i)\neq\bot
		\\ \bot
		\\ \qquad  \mbox{otherwise}
	\end{array}\right.\]
	and for any two atoms $A = a(t^A_1, \dots, t^A_n)$ and $B = b(u^B_1, \dots, u^B_m)$, it holds that
	\[\au^\Phi_\preceq(A,B)=\left\{\begin{array}{l}
		a\big(\au^\Phi_\preceq(t^A_1, u^B_1),\dots, \au^\Phi_\preceq(t^A_n, u^B_n)\big) 
		\\ \qquad \mbox{if } a = b \mbox{ and } n = m 
		\\ \qquad \mbox{and } \forall i \in 1..n: \au^\Phi_\preceq(t^A_i, u^B_i) \neq\bot
		\\ \bot  
		\\ \qquad \mbox{otherwise}
	\end{array}\right.\]
\end{definition}

\begin{example}
	In Table~\ref{table:preceq}, we treat the anti-unification of the same atoms as above, this time with the use of $\au_\preceq^\Phi$ with $\Phi$ a given variabilization function. Note how $\au_\preceq$ behaves differently than $\au_\sqsubseteq$ on the second and third couple of atoms as it requires its arguments to exhibit a similar structure in order to be anti-unifiable.
\end{example}

\begin{table*}
	\caption{Example results for $au_\preceq^\Phi$}
	\label{table:preceq}
	\centering
	\begin{tabular}{l|l|l}
		$\bm{A_1}$ & $\bm{A_2}$ & $\bm{\au_\preceq^\Phi(A_1, A_2)}$\\\hline 
		$p(X, 5, q(Y,4))$ & $p(W,t(Z))$ & $\bot$\\\hline 
		$p(r(X,3), t(5))$ & $p(W, t(Z))$ & $\bot$\\\hline 
		$p(r(X,3), t(Y))$ & $p(r(W,3),t(Z))$ & $p(r(\Phi(X,W),3), t(\Phi(Y,Z)))$ 
	\end{tabular}
\end{table*}

Now, in order to compute $\sqsubseteq$-msgs, we need a more precise anti-unification operator: one that goes deeper into detail when comparing atoms so as not to miss their maximal common structure. 
\begin{definition} 
	Given some variabilization function $\Phi$, let $\dau^\Phi_\sqsubseteq$ (or simply $\dau_\sqsubseteq$ if $\Phi$ is clear from the context) denote the function such that for any two terms $T = t(t_1, \dots, t_n)$ and $U = u(u_1, \dots, u_m)$ it holds that 
	\[\dau^\Phi_\sqsubseteq(T,U)=\left\{\begin{array}{l}
		
		t\big(\dau^\Phi_\sqsubseteq(t_1,u_1), \dots, \dau^\Phi_\sqsubseteq(t_n, u_n)\big) 
		\\ \qquad \mbox{if } t = u \mbox{ and } n = m 
		\\ \qquad \mbox{and } T \notin \mathcal{V} \mbox{ and } U \notin \mathcal{V}
		\\ \Phi(T,U) 
		\\ \qquad \mbox{otherwise}
	\end{array}\right.\]
	and for any two atoms $A = a(t^A_1, \dots, t^A_n)$ and $B = b(u^B_1, \dots, u^B_m)$, it holds that
	\[\dau^\Phi_\sqsubseteq(A,B)=\left\{\begin{array}{l}
		a\big(\dau^\Phi_\sqsubseteq(t^A_1, u^B_1),\dots, \dau^\Phi_\sqsubseteq(t^A_n, u^B_n)\big) 
		\\ \qquad \mbox{if } a = b \mbox{ and } n = m 
		\\ \bot
		\\ \qquad \mbox{otherwise}
	\end{array}\right.\]
\end{definition}

When applied on atoms, it is easy to see that $\dau_\sqsubseteq$ is an anti-unification operator based on relation $\sqsubseteq$.

\begin{example}
	Let us once more consider the anti-unification of the atoms introduced in Example~\ref{ex-au-sq}. This time we make use of $\dau_\sqsubseteq^\Phi$ with $\Phi$ a given variabilization function, to anti-unify the three pairs of atoms. The result is shown in Table~\ref{table:dau}. Notice how the operator preserves as much non-variable atomic structure as possible in the process.
\end{example}
\begin{table*}
	\caption{Example results for $dau_\sqsubseteq^\Phi$}
	\label{table:dau}
	\centering
	\begin{tabular}{l|l|l}
		$\bm{A_1}$ & $\bm{A_2}$ & $\bm{\dau_\sqsubseteq^\Phi(A_1, A_2)}$\\\hline 
		$p(X, 5, q(Y,4))$ & $p(W,t(Z))$ & $\bot$\\\hline 
		$p(r(X,3), t(5))$ & $p(W, t(Z))$ & $p(\Phi(r(X,3),W),t(\Phi(5,Z)))$\\\hline 
		$p(r(X,3), t(Y))$ & $p(r(W,3),t(Z))$ & $p(r(\Phi(X,W),3), t(\Phi(Y,Z)))$ 
	\end{tabular} 
\end{table*}
The existence of these operators proves Lemma~\ref{lemma-au-op}.
	\end{proof}
	
	\section{Proof of Theorem~\ref{thm-ausqsubseteq}}
	\begin{proof}
	Obviously the $\au_\sqsubseteq(A_1,A_2)$ operation can be achieved in a time linear with respect to the arity $n$ of $A_1$. In the worst case, the operation needs to be performed for each atom in $G_1$ with respect to each atom in $G_2$. Hence the first result.

	It is also easy to see that the $\au_\preceq(A_1,A_2)$ operation can be achieved in linear time with respect to the maximum number of function applications in the argument terms of the atom $A_1$ under scrutiny. In the worst case, the operation needs to be performed for each atom in $G_1$ with respect to each atom in $G_2$. Hence the second result.
	\end{proof}
		
%

	\section{Proof of Theorem~\ref{thm-sqsubseteq-msg}}
\begin{proof}
	First note how the atomic anti-unifications and the weights of the associated bipartite graph's edges can be computed simultaneously, by working out $\dau_\sqsubseteq(A_1,A_2)$ for each possible couple $(A_1,A_2)$ in $G_1\times G_2$ and keeping account of the number of non-variable terms encountered during the operation (or $-1$). Given that $\dau_\sqsubseteq(A_1,A_2)$ can obviously operate linearly in the number of terms appearing in $A_1$ (denoted $N$), the computation of all weights is carried out in a time not exceeding $\mathcal{O}(|G_1|.|G_2|.N)$.
	
	Now the obtained assignment problem can be solved by existing algorithms (such as the Hungarian method~\cite{assignment}) that compute a MWM in $\mathcal{O}(n^3)$, where $n$ is the number of vertexes appearing on the side of the bipartite graph that has the most vertexes. In our case, there are $|G_1|$ left vertexes and $|G_2|$ right vertexes so that a MWM algorithm can be ran in $\mathcal{O}(max(|G_1|,|G_2|)^3)$.
\end{proof}
	
	\section{Proof of Theorem~\ref{thm-dataflow-np-complete}}
\begin{proof}
	First, let us consider MSG-MIN. It clearly belongs to NP. Indeed, given an arbitrary generalization $G$, we can verify in polynomial time whether it is a most specific generalization. The procedure is as follows. We can compute at least one $\leqslant$-msg, say $G'$, in polynomial time (see Theorem~\ref{thm-sqsubseteq-msg}). It suffices then to compare the $\tau$-value of $G'$ with that of $G$ in order to decide whether $G$ is a msg. Next, verifying whether the number of variables in $G$ is bounded by a constant is obviously achieved in polynomial time as well.
	
	In order to prove NP-hardness, we will construct a reduction from the well-known set cover problem (known to be NP-complete~\cite{karp}) to MSG-MIN. The set cover problem in its decision-problem version (denoted SCP), can be formulated as follows. Given a constant $p \in \mathbb{N}_0$, a universe $U$ of values and a collection $S$ composed of $n$ sets $\{S_1, \dots, S_n\}$ that cover $U$, i.e. $U = \underset{i=1}{\overset{n}{\cup}}S_i$, the problem is to decide whether there exists $p$ subsets from $S$ that still cover $U$.
	
	We can transform an arbitrary instance of SCP into MSG-MIN as follows. Let us consider without loss of generality a universe $U$ where the elements are lowercase strings and $p \in \mathbb{N}_0$ a constant. Given a collection of sets $S=\{S_1, \dots, S_n\}$ we construct an instance of MSG-MIN as follows. In our construction we use $n+1$ different variables, namely $V$ and $(W_i)_{i\in1..n}$. We use $x_j$ to denote some element of $U$; these elements being strings, we can easily use them as predicate names. The construction of goals $G_1$ and $G_2$ proceeds then as follows:
	
	\begin{algorithmic}
		\State $G_1 = \{\}$ 
		\State $G_2 = \{\}$ 
		\For {each ($S_i \in S$)}
		\For {each ($x_j \in S_i$)}
		\State $G_1 \gets G_1\cup \{x_j(V)\}$				
		\State $G_2 \gets G_2\cup \{x_j(W_i)\}$
		\EndFor
		\EndFor
	\end{algorithmic}
	Note that all the atoms in $G_1$ have the same argument (namely the variable $V$) and there are as many atoms in $G_1$ as there are distinct elements in $S$. In $G_2$, however, there is an atom of the form $x_j(W_i)$ for each element $x_j$ occurring in $S_i$.
	
	The construction is such that any $\leqslant$-msg of $G_1$ and $G_2$ will be a version of $G_1$ where each occurrence of a variable $V$ is replaced by $\Phi(V, W_k)$ for some $W_k\in\vars(G_2)$ (where $\Phi$ is a variabilization function). Now, introducing such a variable $\Phi(V, W_k)$ in the generalization will allow to reuse the same variable for all the atoms $x_j(V)$ in $G_1$ that have a corresponding $x_j(W_k)$ in $G_2$. In other words, choosing to have variable $\Phi(V,W_k)$ in the $\leqslant$-msg is the same as selecting the subset $S_k$ to be part of the solution of the set cover problem. Consequently, using this transformation MSG-MIN can be used to decide SCP. Since the transformation can clearly be done in polynomial time, and since SCP is known to be NP-complete, we conclude that MSG-MIN is NP-complete as well.
	
	Now let us prove the result for LCG-MIN. We know that a $\leqslant$-lcg can be computed in polynomial time, so that a positive instance of LCG-MIN can be verified just like it can be for MSG-MIN. Moreover, the absence of non-variable terms in the transformation from SCP to MSG-MIN above allows us to reuse said transformation as-is to prove that LCG-MIN is NP-hard. Indeed, since the obtained anti-unification problem doesn't harbor terms other than variables, it is both an instance of MSG-MIN and LCG-MIN. LCG-MIN is therefore also NP-complete.
\end{proof}
	
	\section{Proof of Theorem~\ref{thm-inj-np-complete}}
\begin{proof}
	INJ is in NP: given a relation $\leqslant^\iota$, goals $G_1$ and $G_2$ and a substitution (or renaming) $\theta$, it is possible to verify in polynomial time whether the application of $\theta$ on $G_1$ results on a subset of $G_2$ or not.
	As for the proof of NP-hardness, we refer to~\cite{gen} in which the problem ``is $G_1$ a $\preceq^\iota$-lcg of $G_1$ and $G_2$?'' has been proved to be NP-complete using a polynomial reduction from the Induced Subgraph Isomorphism Problem~\cite{SYSLO198291}. The same reduction can be used for the other cases, leading to the conclusion that INJ is NP-complete.
\end{proof}


\end{document}